\documentclass[letterpaper,USenglish,11pt]{article}

\usepackage{comment}
\usepackage{ifthen}
\usepackage{xspace}
\usepackage{float}
\usepackage{subcaption}
\usepackage{amsmath}
\usepackage{amsthm}
\usepackage{amstext}
\usepackage{amssymb}
\usepackage{amsfonts}
\usepackage{bm}
\usepackage[multiple]{footmisc}
\usepackage{wrapfig}
\usepackage{url}
\usepackage[margin=1in]{geometry}
\usepackage{mathdots}
\usepackage{mathtools}
\usepackage{thmtools}
\usepackage{thm-restate}
\usepackage{tikz}
\usetikzlibrary{calc, fit, shapes,intersections,arrows,matrix, topaths}
\usetikzlibrary{decorations,decorations.pathreplacing,patterns, shadows}
\usepackage{graphicx}
\usepackage[linesnumbered,noresetcount,vlined,ruled]{algorithm2e}
\usepackage[inline]{enumitem}
\usepackage{environ}
\usepackage{dsfont}
\usepackage{cite}
\usepackage{yfonts} 
\usepackage{graphicx}

\newcommand{\FP}{\text{FP}}
\newenvironment{proofsketch}{%
  \proof}{\endproof}
\renewcommand{\epsilon}{\varepsilon} 

\renewcommand{\epsilon}{\varepsilon}
\newcommand{\eps}{\varepsilon}
\newcommand{\Dict}{\textsf{Dict}}
\newcommand{\MSDict}{\textsf{MS-Dict}}
\newcommand{\Ret}{\textsf{Ret}}

\usepackage{amsmath} \usepackage{amsthm}
\newtheorem{theorem}{Theorem}
\newtheorem{lemma}[theorem]{Lemma}

\newtheorem{definition}[theorem]{Definition}

\newtheorem{claim2}[theorem]{Claim}
\newtheorem{observation}[theorem]{Observation}

\newcommand{\UU}{\mathcal{U}}
\newcommand{\hUU}{\hat{\mathcal{U}}}

\newcommand{\UUh}{\hat{\mathcal{U}}}
\newcommand{\uh}{\hat{u}}
\newcommand{\MM}{\mathcal{M}}

\newcommand{\hb}{h^b}

\newcommand{\DD}{\mathcal{D}}
\newcommand{\etal}{\textit{et al.}\xspace}
\newcommand{\size}[1]{\ensuremath{\left|#1\right|}}
\newcommand{\set}[1]{\left\{ #1 \right\}}
\newcommand{\parentheses}[1]{\left(#1\right)}

\DeclarePairedDelimiter{\ceil}{\lceil}{\rceil}
\newcommand{\expectation}[2]{\mathbb{E}_{#1}\left[ #2 \right]}
\renewcommand{\Pr}[1]{{\mathrm{Pr}}\left[ #1 \right]}
\newcommand{\NN}{\mathbb{N}}
\DeclareMathOperator{\query}{\textsf{query}}

\DeclareMathOperator{\ins}{\textsf{insert}}

\DeclareMathOperator{\del}{\textsf{delete}}

\DeclareMathOperator{\op}{\textsf{op}}
\DeclareMathOperator{\out}{\textsf{out}}

\DeclareMathOperator{\hheader}{\textsf{header}}
\DeclareMathOperator{\bbody}{\textsf{body}}

\DeclareMathOperator{\poly}{\textsf{poly}}
\newcommand{\ind}[1]{\mathds{1}_{#1}}
\newcommand{\indD}{\ind{\DD(t)}}
\makeatother

\begin{document}
\title{A Dynamic Space-Efficient Filter with Constant Time Operations\thanks{The conference
version of this paper has appeared in SWAT 2020. An earlier
version of this paper appears in~\cite{bercea2019fullydynamic}. The construction presented
in~\cite{bercea2019fullydynamic} is from first principles. In particular, it   does not employ reductions from the retrieval problem in the sparse case.}}
\author{
 Ioana O. Bercea\thanks{
Tel Aviv University, Tel Aviv, Israel.
Email:~\texttt{ioana@cs.umd.edu, guy@eng.tau.ac.il}. This research was supported by a
grant from the United States-Israel Binational Science Foundation
(BSF), Jerusalem, Israel, and the United States National Science
Foundation (NSF).}
\and
Guy Even\footnotemark[2]
}
\date{}
\maketitle

\begin{abstract} 
  A dynamic dictionary is a data structure that maintains sets of
  cardinality at most $n$ from a given universe and supports
  insertions, deletions, and membership queries. A filter approximates
  membership queries with a one-sided error that occurs with
  probability at most $\eps$. The goal is to obtain dynamic filters
  that are space-efficient (the space is $1+o(1)$ times the
  information-theoretic lower bound) and support all operations in
  constant time with high probability.  One approach to designing
  filters is by a reduction to the retrieval problem.  When the size of the
  universe is polynomial in $n$, this approach yields a
  space-efficient dynamic filter as long as the error parameter $\eps$
  satisfies $\log(1/\eps) = \omega(\log\log n)$.
  For the case that $\log(1/\eps) = O(\log\log n)$, we present the
  first space-efficient dynamic filter with constant time operations
  in the worst case (whp).  In contrast, the space-efficient dynamic
  filter of Pagh~\etal~\cite{DBLP:conf/soda/PaghPR05} supports
  insertions and deletions in amortized expected constant time.  Our
  approach employs the classic reduction of
  Carter~\etal~\cite{carter1978exact} on a new type of dictionary
  construction that supports random multisets.

\end{abstract}

\section{Introduction}
We consider the problem of maintaining datasets
subject to insert, delete, and membership query operations. 
Given a set $\DD$ of $n$ elements from a universe $\UUh$, a
membership query asks if the queried element $x \in \UUh$
belongs to the set $\DD$. When exact answers are required, the
associated data structure is called a \emph{dictionary}. When
one-sided errors are allowed, the associated data structure is called
a \emph{filter}. Formally, given an error parameter $\eps>0$, a filter
always answers ``yes'' when $x \in \DD$, and when $x\notin \DD$, it
makes a mistake with probability at most $\eps$. We refer to such an
error as a \emph{false positive} event%
\footnote{The probability is taken over the random choices that the
filter makes.}.

When false positives can be tolerated, the main advantage of using a
filter instead of a dictionary is that the filter requires much less
space than a dictionary~\cite{carter1978exact,lovett2010lower}.  Let
$\uh\triangleq \size{\UUh}$ be the size of the universe and $n$ denote
an upper bound on the size of the set at all points in time. The
information theoretic lower bound for the space of dictionaries is
$\ceil{\log_2 {\uh \choose n}} = n\log(\uh/n)+ \Theta(n)$
bits.\footnote{All logarithms are base $2$ unless otherwise
  stated. $\ln x$ is used to denote the natural
  logarithm.}\footnote{ This equality
  holds when $\uh$ is significantly larger than $n$.} On the other hand,
the lower bound for the space of filters is $n\log(1/\eps)$
bits~\cite{carter1978exact}. In light of these lower bounds, we call a
dictionary \textit{space-efficient} when it requires
$(1+o(1))\cdot n\log(\uh/n))+ \Theta(n)$ bits, where the term $o(1)$
converges to zero as $n$ tends to infinity. Similarly, a
space-efficient filter requires $(1+o(1))\cdot n\log(1/\eps) + O(n)$
bits.\footnote{An asymptotic expression that mixes big-O and small-o
  calls for elaboration. If $\eps=o(1)$, then the asymptotic
  expression does not require the $O(n)$ addend. If $\eps$ is
  constant, the $O(n)$ addend only emphasizes the fact that the
  constant that multiplies $n$ is, in fact, the sum of two constants:
  one is almost $\log (1/\eps)$, and the other does not depend on
  $\eps$. Indeed, the lower bound in~\cite{lovett2010lower} excludes
  space $(1+o(1))\cdot n\log(1/\eps)$ in the dynamic setting for
  constant values of $\eps$.}

When the set $\DD$ is fixed, we say that the data structure is
\emph{static}. When the data structure also supports insertions, we
say that it is \emph{incremental}. Data structures that handle both
deletions and insertions are called \emph{dynamic}.

The goal is to design dynamic dictionaries and filters that achieve
``the best of both worlds''~\cite{arbitman2010backyard}: they are
space-efficient and perform operations in constant time in the worst
case with high probability.\footnote{By with high probability (whp),
  we mean with probability at least $1-1/n^{\Omega(1)}$. The constant
  in the exponent can be controlled by the designer and only affects
  the $o(1)$ term in the space of the dictionary or the filter.}

\medskip\noindent \textbf{The Dynamic Setting.}  One approach for
designing dynamic filters was suggested by
Pagh~\etal~\cite{DBLP:conf/soda/PaghPR05}, outlined as follows.
Static (resp., incremental) filters can be obtained from static
(resp., incremental) dictionaries for sets by a reduction of
Carter~\etal~\cite{carter1978exact}.  This reduction simply hashes the
universe to a set of cardinality $n/\eps$.  Due to collisions, this
reduction does not directly lead to dynamic filters.  Indeed, if two
elements $x$ and $y$ in the dataset collide, and $x$ is deleted, how
is $y$ kept in the filter?  To overcome the problem with deletions, an
extension of the reduction to the dynamic setting was proposed by
Pagh~\etal~\cite{DBLP:conf/soda/PaghPR05}. This proposal is based on
employing a dictionary that maintains multisets rather than sets
(i.e., elements in multisets have arbitrary multiplicities). This
extension combined with a dynamic dictionary for multisets yields a
dynamic filter~\cite{DBLP:conf/soda/PaghPR05}. In fact, Pagh~\etal
obtain a dynamic filter that is space-efficient but performs
insertions and deletions in amortized constant time (but not in the
worst case).  Until recently, the design of a dynamic dictionary on
multisets that is space-efficient and performs operations in constant
time in the worst case whp was open~\cite{arbitman2010backyard}. In
this paper, we avoid the need for supporting arbitrary multisets by
observing that it suffices to support random multisets (see
Sec.~\ref{sec:reduce}).\footnote{We recently resolved the problem of
  supporting arbitrary multisets in~\cite{bercea2020spaceefficient} (thus the dictionary
  in~\cite{bercea2020spaceefficient} can support arbitrary multisets vs. the dictionary
  presented here that only supports random multisets).}

Another approach for designing filters employs retrieval data
structures. In the retrieval problem, we are given a function
$f:\DD\rightarrow \set{0,1}^k$, where $f(x)$ is called the
\emph{satellite data} associated with $x\in\DD$. When an element
$x\in\UUh$ is queried, the output $y$ must satisfy $y=f(x)$ if
$x\in\DD$ (if $x\notin\DD$, any output is allowed). By storing a
random fingerprint of length $\log(1/\eps)$ as satellite data, a
retrieval data structure can be employed as a filter at no additional
cost in space and with an $O(1)$ increase in time per
operation~\cite{dietzfelbinger2008succinct,porat2009optimal} (the
increase in time is for computing the fingerprint). This reduction was
employed in the static case and it also holds in the dynamic case (see
Sec.~\ref{sec:sparse}).

Using the dynamic retrieval data structure of
Demaine~\etal~\cite{demaine2006dictionariis}, one can obtain a filter
that requires $(1+o(1))\cdot n\log(1/\eps) + \Theta(n\log\log(\uh/n))$
bits and performs operations in constant time in the worst case whp.
When the size of the universe satisfies $\uh=\poly(n)$, this reduction
yields a space-efficient filter when the false positive probability
$\eps$ satisfies $\log(1/\eps) = \omega(\log \log n)$ (which we call
the \emph{sparse case}).\footnote{The terms ``sparse'' and ``dense''
  stem from the fact that the reduction of
  Carter~\etal~\cite{carter1978exact} is employed.  Thus, the
  filter is implemented by a dictionary that stores $n$ elements from
  a universe of cardinality $n/\eps$.}  This approach is inherently
limited to the sparse case since dynamic retrieval data structures
have a space lower bound of $\Theta(n\log\log(\uh/n))$ regardless of
the time each operation takes and even when storing two bits of
satellite data~\cite{mortensen2005dynamic}.

Thus, the only case in which a space-efficient dynamic filter with
constant time operations is not known is when
$\log(1/\eps) = O(\log \log n)$.  We refer to this case as the
\emph{dense case}.  The dense case occurs, for example, in
applications in which $n$ is large and $\eps$ is a constant (say
$\eps=1\%$).

\subsection{Our Contributions}

In this paper, we present the first dynamic space-efficient
filter for the dense case with constant time operations in the worst case whp. 
In the following theorem, we assume that the size of the universe $\UUh$ is
polynomial in $n$.\footnote{This is justified by mapping $\UUh$ to
$[\poly(n)]$ using $2$-independent hash
functions~\cite{demaine2006dictionariis}.}
We allow $\eps$ to be as small as $n/|\UUh|$ (below this threshold,
simply use a dictionary).  Memory accesses are in the RAM model in
which every memory access reads/writes a word of $\Theta(\log n)$
contiguous bits.  All computations we perform over one word take
constant time (see Sec.~\ref{sec:bin}). \emph{Overflow} refers to the event that the space
allocated for the filter does not suffice.

\begin{theorem}\label{thm:filter}
There exists a dynamic filter that maintains a set of at most
$n$ elements from a universe $\UUh=[\uh]$, where $\uh=\poly(n)$ with the
following guarantees: (1)~For every polynomial in $n$ sequence of
insert, delete, and query operations, the filter does not overflow
whp. (2)~If the filter does not overflow, then every operation
(query, insert, and delete) can be completed in constant
time. (3)~The required space is $(1+o(1))\cdot n\log (1/\eps)+O(n)$
bits. (4)~For every query, the probability of a false positive event
is bounded by $\eps$.
\end{theorem}
Our result is based on the observation that it suffices to use the
reduction of Carter~\etal~\cite{carter1978exact} on dictionaries that
support \emph{random} multisets rather than arbitrary multisets. A
random multiset is a uniform random sample (with replacements) of
the universe.  In Sec.~\ref{sec:reduce}, we prove that the reduction
of Carter~\etal~\cite{carter1978exact} can be applied in this new
setting.  We then design a dynamic space-efficient dictionary that
works on random multisets from a universe $\UU=[u]$ with
$\log(u/n)=O(\log\log n)$ (Sec.~\ref{sec:rmsdic}). The dictionary
supports operations in constant time in the worst case whp.  Applying
the reduction of Carter~\etal~\cite{carter1978exact} to this new
dictionary yields our dynamic filter in the dense case. Together with
the filter construction for the sparse case (included, for
completeness, in Sec.~\ref{sec:sparse}), we obtain
Theorem~\ref{thm:filter}.

\subsection{Our Model}
\medskip\noindent \textbf{Memory Access Model.} We assume that the
data structures are implemented in the RAM model in which the basic
unit of one memory access is a word. Let $w$ denote the memory word
length in bits. We assume that $w=\Theta(\log n)$.  See
Sec.~\ref{sec:bin} for a discussion of
how computations over words are implemented in constant time.

\medskip\noindent
\textbf{Success Probability.}
We prove that overflow occurs with
probability at most $1/\poly(n)$ and that one can control the degree
of the polynomial (the degree of the polynomial only
affects the $o(1)$ term in the size bound). In the case of random multisets, the
probability of an overflow is a joint probability distribution over
the random choices of the dictionary and the distribution over the
realizations of the multiset. In the case of sets, the probability of
an overflow depends only on the random choices that the filter
makes.

\medskip\noindent
\textbf{Hash Functions.}
The filter for the dense case employs pairwise independent hash functions and
invertible permutations of the universe that can be
evaluated in constant time and that have a small space representation (i.e., the one-round Feistel permutations of Arbitman~\etal~\cite{arbitman2010backyard} or the
quotient hash functions of Demaine~\etal~\cite{demaine2006dictionariis}). 
For simplicity, we first analyze the filter construction assuming fully random hash functions (Sec.~\ref{sec:ovf}).
In Sec.~\ref{sec:hash}, we prove that the same arguments hold when we
use succinct hash functions.  

\medskip\noindent
\textbf{Worst Case vs. Amortized.}
An interesting application that emphasizes the importance of
worst-case performance is that of handling search engine queries. Such
queries are sent in parallel to multiple servers, whose responses are
then accumulated to generate the final output. The latency of this
final output is determined by the slowest response, thus reducing the
average latency of the final response to the worst latency among the
servers. See~\cite{broder2001using,kirsch2007using, arbitman2009amortized,arbitman2010backyard}
for  further discussion on the shortcomings of expected or amortized
performance in practical scenarios.

\medskip\noindent \textbf{The Extendable Setting.} This paper deals
with the non-extendable setting in which the bound $n$ on the
cardinality of the dataset is known in advance.  The filter is
allocated space that is efficient with respect to the lower bound on
the space of a filter with parameters $u,n,\eps$ (corresponding to the
universe's cardinality, the maximum cardinality of the set, and the
error parameter). The extendable scenario in which space must adapt to
the current cardinality of the dataset is addressed in
Pagh~\etal~\cite{pagh2013approximate}. In fact, they prove that
extendible filters require an extra $\Omega(\log \log n)$ bits per
element.

\subsection{Related Work}
The topic of dictionary and filter design is a fundamental
theme in the theory and practice of data structures.
We restrict our focus to the results that
are closest to our setting (i.e., are space-efficient, take constant
time per operation, support dynamic sets).

\medskip\noindent
\textbf{Dictionaries.}
The dictionary of Arbitman~\etal~\cite{arbitman2010backyard} is the
only space-efficient dynamic dictionary for sets that performs
all operations in constant time in the worst case with high probability.
They leave it as an open question whether one can design a dictionary on multisets
with similar guarantees. Indeed, their construction does not seem to extend
even to the case of random multisets. The main reason is that the second level
of their dictionary (the backyard), implemented as a de-amortized cuckoo hash table, does not
support duplicate elements. Moreover, the upper bound on the number of elements that the backyard stores is $\Omega\parentheses{\frac{\log\log n}{(\log n)^{1/3}}\cdot n}$. As such, 
it cannot accommodate storing naive fixed-length counters of elements (which would require $\Theta(\log n)$ bits per element) without
rendering the dictionary space-inefficient. 

The space-efficient dynamic dictionary for multisets of Pagh~\etal~\cite{DBLP:conf/soda/PaghPR05} 
supports queries in constant time, and insertions/deletions in amortized expected constant time.
For dictionaries on sets,
several dynamic constructions 
support operations in constant time with high
probability but are not
space-efficient~\cite{ dietzfelbinger1990new, dalal2005two,
demaine2006dictionariis,
arbitman2009amortized}. On the other hand, some dictionaries are
space-efficient but do not have constant time guarantees with high
probability for all of their 
operations~\cite{raman2003succinct, fotakis2005space,panigrahy2005efficient,dietzfelbinger2007balanced}. For the static case, several space-efficient constructions exist
that perform queries in constant time~\cite{brodnik1999membership,pagh2001low,patrascu2008succincter,yu2019nearly}.

\medskip\noindent
\textbf{Filters.} 
The filters of Pagh
\etal~\cite{DBLP:conf/soda/PaghPR05} and Arbitman~\etal~\cite{arbitman2010backyard} follow
from their respective dictionaries by employing the reduction of Carter~\etal~\cite{carter1978exact}.
Specifically, the dynamic filter of Pagh
\etal~\cite{DBLP:conf/soda/PaghPR05} supports  queries in constant time
and insertions and deletions in amortized expected constant time. The
incremental filter of Arbitman~\etal~\cite{arbitman2010backyard}
performs queries in constant time and insertions in constant time with
high probability. It does not support deletions.

The construction of Bender~\etal~\cite{bender2018bloom} describes a
dynamic adaptive filter that assumes access to fully random hash
functions.~\footnote{Loosely speaking, an adaptive filter is one that
  fixes false positives after they occur~\cite{bender2018bloom,
    mitzenmacher2018adaptive}.} The adaptive filter works in
conjunction with an external memory dictionary (on the set of
elements) and supports operations in constant time with high
probability (however, an insert or query operation may require
accessing the external memory dictionary). The space of the external
memory dictionary is not counted in the space of their filter. The
(in-memory) filter they employ is a variant of the dynamic quotient
filter~\cite{Cleary, DBLP:conf/soda/PaghPR05,bender2012thrash,
  DBLP:conf/sigmod/PandeyBJP17}. The space-efficient quotient filter
employs linear probing and performs operations only in expected
constant time for large values of
$\eps$~\cite{DBLP:conf/soda/PaghPR05}.  The filter
in~\cite{bender2018bloom} tries to avoid a large running time per
insert operation by bounding the displacement of the inserted
element. Hence, if (Robin Hood) linear probing does not succeed after
a constant number of words, then the element is inserted in a
secondary structure (see Sec. $5.3$ in~\cite{bender2017bloom}). There
is a gap in~\cite{bender2017bloom} regarding the question of whether
bounded displacements guarantee constant time operations in the worst
case. Specifically, searching for an element requires finding the
beginning of the ``cluster'' that contains the ``run'' associated with
that particular element. No description or proof is provided
in~\cite{bender2017bloom} that the beginning of the cluster is a
constant number of words away from the ``quotient'' in the worst case.

Other filters of interest include the dynamic filter of
Pagh~\etal~\cite{pagh2013approximate} that adjusts its space on the
fly to the cardinality of the dataset (hence, works without knowing
the size of the dataset in advance) and performs operations in
constant time.  Pagh~\etal~\cite{pagh2013approximate} also prove a
lower bound that forces a penalty of $O(\log \log n)$ per element for
such ``self-adjusting'' dynamic filters. Another filter is the cuckoo
filter, whose performance  depends on the number of
elements currently stored in the filter but that has been reported to work
well in
practice~\cite{DBLP:conf/conext/FanAKM14,DBLP:conf/swat/Eppstein16}. Space-efficient
filters for the static case have been studied
extensively~\cite{mitzenmacher2002compressed,dietzfelbinger2007balanced,dietzfelbinger2008succinct,
  porat2009optimal}.

\subsection{Paper Organization}
Preliminaries are in Sec.~\ref{sec:prelim}. The proof that the reduction of Carter~\etal~\cite{carter1978exact} can
be employed to construct dynamic filters from dynamic dictionaries
on random multisets can be found in Sec.~\ref{sec:reduce}. The filter for the dense
case is described and analyzed in Sec.~\ref{sec:rmsdic}. Section~\ref{sec:hash}
includes a discussion on how to remove the assumption of access to fully random hash functions from
Sec.~\ref{sec:ovf}. Section~\ref{sec:sparse} reviews the construction of a filter in the sparse
case based on a retrieval data structure.
 Theorem~\ref{thm:filter} is proved in Sec.~\ref{sec:prooffilter}.

\section{Preliminaries}\label{sec:prelim}

\medskip\noindent
\textbf{Notation.}
The \emph{indicator function} of a set $S$ is the function
$\mathds{1}_S: S\rightarrow \set{0,1}$ defined by
\begin{align*}
\mathds{1}_S(x)&\triangleq
\begin{cases}
1 & \text{if $x\in S$},\\
0 & \text{if $x\not\in S$}\;.
\end{cases}
\end{align*}
For any positive $k$, let $[k]$ denote the set
$\set{0,\ldots,\ceil{k}-1}$.  For a string $a \in \set{0,1}^*$, let $\size{a}$ denote the length
of $a$ in bits.

The binary representation of a natural number bijectively maps the set $[k]$ to the set
$\{0,1\}^{\ceil{\log_2 k}}$. Hence, for a hash function $h$ whose
range is $[k]$, we often treat $h(x)$ as a binary string of length
$\ceil{\log_2 k}$.

\subsection{Filter and Dictionary Definitions}
Let $\UUh$ denote the universe of all possible elements.

\medskip\noindent
\textbf{Operations.} We consider three types of operations:
\begin{itemize}
\item $\ins(x_t)$ - insert $x_t\in \UUh$ to the dataset.
\item $\del(x_t)$ - delete $x_t\in \UUh$ from the dataset.
\item $\query(x_t)$ - is $x_t\in \UUh$ in the dataset?
\end{itemize}

\medskip\noindent
\textbf{Dynamic Sets and Random Multisets.}
Every sequence of operations $R=\set{\op_{t}}_{t=1}^{T}$ defines a
\emph{dynamic set} $\DD(t)$ over $\UUh$ as follows.\footnote{ The
definition of $\DD(t)$ in Equation~\ref{eq:state} does not rule out a
deletion of $x\notin\DD(t-1)$.  However, we assume that
$op_t=\del(x_t)$ only if $x_t\in \DD(t-1)$.
}

\begin{align}\label{eq:state}
\DD(t)&\triangleq
\begin{cases}
\emptyset & \text{if $t=0$}\\
\DD(t-1)\cup \set{x_t} &\text{if $\op_t=\ins(x_t)$}\\
\DD(t-1)\setminus \set{x_t} &\text{if $\op_t=\del(x_t)$}\\
\DD(t-1)& \text{if $t>0$ and $\op_t=\query(x_t)$.}
\end{cases}
\end{align}

\begin{definition}
A \emph{multiset} $\MM$ over $\UUh$ is a function $\MM:\UUh\rightarrow
\NN$. We refer to $\MM(x)$ as the \emph{multiplicity} of $x$. If
$\MM(x)=0$, we say that $x$ is not in the multiset. We refer to
$\sum_{x\in \UUh} \MM(x)$ as the \emph{cardinality} of the multiset
and denote it by $\size{\MM}$.
\end{definition}

The \emph{support} of the multiset is the set
$\set{x \mid \MM(x)\neq 0}$. The \emph{maximum multiplicity} of a multiset is
$\max_{x\in\UUh} \MM(x)$.

A \emph{dynamic multiset} $\set{\MM_t}_t$ is specified by a sequence
of insert and delete operations. Let $\MM_t$ denote the multiset after
$t$ operations.\footnote{As in the case of dynamic sets, we allow $\op_t=\del(x_t)$ only if $\MM_{t-1}(x_t)>0$.}
\begin{align*}
\MM_t(x)&\triangleq
\begin{cases}
0 & \text{if $t=0$}\\
\MM_{t-1}(x)+1 &\text{if $\op_t=\ins(x)$}\\
\MM_{t-1}(x)-1 &\text{if $\op_t=\del(x)$}\\
\MM_{t-1}(x)& \text{otherwise.}
\end{cases}
\end{align*}
\medskip\noindent We say that a dynamic multiset $\set{\MM_t}_t$ has
cardinality at most $n$ if $\size{\MM_t}\leq n$, for every $t$.
\begin{definition}
A dynamic multiset $\MM$ over $\UUh$ is a\emph{ random
multiset} if for every $t$, the multiset $\MM_t$ is the outcome of 
independent uniform samples (with replacements) from $\UUh$.
\end{definition}

\medskip\noindent
\textbf{Dynamic Filters.}
A \emph{dynamic filter} is a data structure that maintains a dynamic set
$\DD(t)\subseteq \UUh$ and is parameterized by an error parameter
$\eps\in (0,1)$.  Consider an input sequence that specifies a dynamic
set $\DD(t)$, for every $t$.  The filter outputs a bit for every query
operation. We denote the output that corresponds to $\query(x_t)$ by
$\out_t\in\set{0,1}$.  We require that the output satisfy the
following condition:
\begin{align}
\op_t&=\query(x_t) \Rightarrow \out_t \geq \indD(x_t)\;.
\end{align}
The output $\out_t$ is an approximation of $\indD(x_t)$ with a one-sided
error.  Namely, if $x_t\in \DD(t)$, then $b_t$ must equal $1$.

\begin{definition}[false positive event]
Let $FP_t$ denote the event that $\op_t=\query(x_t)$,
$\out_t=1$ and $x_t\notin\DD(t)$.
\end{definition}

\medskip\noindent
The error parameter $\eps\in (0,1)$ is used to bound the probability
of a false positive error.
\begin{definition}
We say that the \emph{false positive probability} in a filter is
bounded by $\eps$ if it satisfies the following property. For every
sequence $R$ of operations and every $t$,

$\Pr{\FP_t}
\leq\eps\;.$
\end{definition}
The probability space in a filter is induced only by the random
choices (i.e., choice of hash functions) that the filter makes. Note
also that if $\op_t=\op_{t'}=\query(x)$, where
$x\not\in \DD(t)\cup \DD(t')$, then the events $\FP_{t}$ and
$\FP_{t'}$ may not be independent
(see~\cite{bender2018bloom,mitzenmacher2018adaptive} for a discussion
of repeated false positive events and adaptivity).

\medskip\noindent \textbf{Dynamic Dictionaries.}  A \emph{dynamic
  dictionary} 
is a dynamic filter with error parameter $\eps=0$. In the case of multisets, the response $\out_t$ of a
dynamic dictionary to a $\query(x_t)$ operation must satisfy
$\out_t=1$ iff $\MM_t(x_t)>0$.\footnote{ One may also define
  $\out_t=\MM_t(x_t)$.}

When we say that a filter or a dictionary has a capacity parameter $n$, we mean
that the cardinality of the input set/multiset is at most $n$ at all
points in time.

\medskip\noindent
\textbf{Success Probability and Probability Space.} We say that a
dictionary (filter) \emph{works} for sets and random multisets if the
probability that the dictionary does not overflow is high (i.e., it is
$\geq 1-1/\poly(n)$). The probability in the case of random multisets
is taken over both the random choices of the dictionary and the
distribution of the random multisets. In the case of sets, the success
probability depends only on the random choices of the dictionary.

\medskip\noindent
\textbf{Dense vs. Sparse.} We differentiate between two cases in the design of filters, depending on $1/\eps$. 
\begin{definition}\label{def:dense}
The \emph{dense case} occurs when $\log(1/\eps)=O(\log \log n)$.  The
\emph{sparse case} occurs when $\log(1/\eps)=\omega(\log\log n)$.
\end{definition}

\section{Reduction: Filters Based on Dictionaries}\label{sec:reduce}
In this section, we employ the reduction of
Carter~\etal~\cite{carter1978exact} to construct dynamic filters out
of dynamic dictionaries for random multisets. Our reduction can be
seen as a relaxation of the reduction of
Pagh~\etal~\cite{DBLP:conf/soda/PaghPR05}. Instead of requiring that
the underlying dictionary support multisets, we require that it only
supports random multisets. We say that a function $h:A\rightarrow B$
is \emph{fully random} if $h$ is sampled u.a.r. from the set of all
functions from $A$ to $B$.

\begin{claim2}
  Consider a fully random hash function
  $h:\UUh \rightarrow \left[\frac{n}{\eps}\right]$ and let
  $\DD\subseteq \UUh$. Then $h(\DD)$ is a random multiset of
  cardinality $\size{\DD}$.
\end{claim2}

Consider a dynamic set $\DD(t)$ specified by a sequence of insert and
delete operations. Since $h$ is random, an ``adversary'' that
generates the sequence of insertions and deletions for $\DD(t)$
becomes an oblivious adversary with respect to $h(\DD(t))$ in the
following sense. Insertion of $x$ translates to an insertion of $h(x)$
which is a random element (note that $h(x)$ may be a duplicate of a
previously inserted element\footnote{Duplicates in $h(\DD(t))$ are
  caused by collisions (i.e., $h(x)=h(y)$) rather than by
  reinsertions.}). When deleting at time $t$, the adversary specifies
a previous time $t'<t$ in which an insertion took place, and requests
to delete the element that was inserted at time $t'$.

Let $\Dict$ denote a dynamic dictionary for random multisets of
cardinality at most $n$ from the universe
$\left[\frac{n}{\eps}\right]$.
\begin{lemma}\label{lemma:reduce dynamic}
For every dynamic set $\DD(t)$ of cardinality at most $n$, a
dictionary $\Dict$ with respect to the random multiset $h(\DD(t))$
and universe $\left[\frac{n}{\eps}\right]$ is a dynamic filter for
$\DD(t)$ with capacity parameter $n$ and error parameter $\eps$.
\end{lemma}
\begin{proofsketch}
A dictionary $\Dict$ records the multiplicity
of $h(x_t)$ in the multiset $h(\DD(t))$ and so deletions are performed correctly. The filter outputs $1$ if
and only if the multiplicity of $h(x_t)$ is positive. False positive events are caused by
collisions in $h$. Therefore, the
probability of a false positive is bounded by $\eps$ because of the
cardinality of the range of $h$.
\end{proofsketch}

\section{Fully Dynamic Filter (Dense Case)}\label{sec:rmsdic}
In this section, we present a fully dynamic filter for the dense case, i.e., 
$\log (1/\eps)=O(\log \log n)$. The reduction in
Lemma~\ref{lemma:reduce dynamic} implies that it suffices to construct
a dynamic dictionary for random multisets. We refer to this dictionary
as the \emph{RMS-Dictionary} (RMS - Random Multi-Set).

The RMS-Dictionary is a dynamic space-efficient dictionary for
random multisets of cardinality at most $n$ from a universe $\UU=[u]$,
where $u=n/\eps$. The dense case implies that
$\log(u/n)=O(\log\log n)$.

The RMS-Dictionary consists of two levels of dictionaries: a set of
\emph{bin dictionaries} (in which most of the elements are stored) and
a \emph{spare} (which stores $n_s=O(n/\log^3 n)$ elements).

\paragraph*{Parametrization.}
We use the following parameters in the construction of the
RMS-dictionary for $n$ elements.
\begin{align*}
  B'&\triangleq \frac{\ln n}{\ln (1+u/n)}=O(\log n)\\
  \delta&\triangleq \frac{\log \log n}{\sqrt{B}}=o(1)\\
  B&\triangleq \frac{\ln 2}{4(1+\delta)}\cdot B'\\
  m&\triangleq \frac{n}{B}\;.
\end{align*}
Let $m$ denote the number of bin dictionaries. So, in expectation,
each bin stores (at most) $B$ elements. To accommodate deviations from
the expectation, extra capacity is allocated in the bin
dictionaries. Namely, each bin dictionary can store up to
$(1+\delta)\cdot B$ elements.

The universe of the spare dictionary is $[u]$. However, the universe
of each bin dictionary is $[u/m]$. The justification for the reduced
universe of bin dictionaries is that the index of the bin contains
$\log m$ bits of information about the elements in it (this reduction
in the universe size is often called
``quotienting''~\cite{Knuth, pagh2001low, DBLP:conf/soda/PaghPR05,demaine2006dictionariis,bender2012thrash}).

\subsection{The Bin Dictionary}\label{sec:bin}
The bin dictionary is a (deterministic) dynamic dictionary for (small)
multisets. Let $u'$ denote the cardinality of the universe from which
elements stored in the bin dictionary are taken. Let $n'$ denote an
upper bound on the cardinality of the dynamic multiset stored in a bin
dictionary (i.e., $n'$ includes multiplicities).  The bin dictionary
must be space-efficient, namely, it must fit in
$n'\log(u'/n') + O(n')$ bits, and must support queries, insertions,
and deletions in constant time.

The specification of the bin dictionary is even more demanding than
the dictionary we are trying to construct. The point is that we focus on
parametrizations in which the bin dictionary fits in a constant number
of words.  Let
$B\triangleq \Theta\parentheses{\frac{\log n}{\log (u/n)}}$ and
$\delta\triangleq \Theta\parentheses{\frac{\log\log
    n}{\sqrt{B}}}=o(1)$.
Recall that the number of bins is $m=n/B$.
\begin{observation}
Let $u'=u/m$ and $n'=(1+\delta)\cdot B$.  The bin
dictionary for $u'$ and $n'$ fits in $O(\log n)$ bits, and hence
fits in a constant number of words.
\end{observation}

We propose two
implementations of bin dictionaries that meet the specifications; one is based on lookup tables,
and the other on Elias-Fano
encoding~\cite{elias1974efficient,fano1971number} (see also Carter~\etal~\cite{carter1978exact}). 
The space
required by the bin dictionaries that employ global lookup tables meets the
information-theoretic lower bound. The space required by the Elias-Fano
encoding is within half a bit per element more than the information-theoretic
 lower bound~\cite{elias1974efficient}.

\medskip\noindent
\textbf{Global Tables. }
We follow Arbitman~\etal~\cite{arbitman2010backyard} and employ a
global lookup table common to all the bin dictionaries.
For the sake of simplicity, we discuss how insertion operations are
supported. An analogous construction works for queries and deletions.

The bin dictionary has $s\triangleq\binom{u'+n'}{n'}$
states.\footnote{Since the bin dictionary stores a multiset of cardinality at most $n'$,
the number of states is actually $\binom{u'+1+n'}{n'}$.} Hence,
we need to build a table that encodes a function
$f:[s]\times [u'] \rightarrow [s]$, such that given a state $i\in [s]$ and
an element $x\in [u']$, $f(i,x)$ encodes the state of the bin
dictionary after $x$ is inserted.  The size of the table that stores
$f$ is $s\cdot u'\cdot \log s$ bits.

We choose the following parametrization so that the table size is
$o(n)$ (recall that $n$ is the upper bound on the cardinality of the
whole multiset).  We set  
$B=\frac{\ln(2)}{4(1+\delta)}\cdot \frac{\ln n}{\ln (1+u/n)}$ (recall
that $B$ is the expected occupancy of a bin). Recall that $u'=\frac{u}{m}=\frac{u}{n/B}$
and $n'=(1+\delta)\cdot B$.  Assume that $1+u/n>e$.  Hence,
\begin{align*}
  s=\binom{u'+n'}{n'}&\leq \parentheses{\frac{e(u'+n')}{n'}}^{n'} = \exp(n'\cdot \log(1+\frac{u'}{n'}) + n')\\
&\leq \exp(n'\cdot \log(1+\frac{u}{n}) + n')\leq \exp(2n'\cdot \log(1+\frac{u}{n}))\\
					 &=\exp(2(1+\delta)\cdot B \cdot \log(1+\frac{u}{n}))\\
					& = \exp(\ln(n)/2) = \sqrt{n}
\end{align*}
Since $u'=u/m \leq \poly(\log n)$ and $\log s=O(\log n)$, we
conclude that the space required to store $f$ is $o(n)$ bits.

Operations are supported in constant time since the table is
addressed by the encoding of the current state and operation.

\medskip\noindent \textbf{Elias-Fano Encoding.}
In this section, we overview a bin dictionary implementation that
employs (a version of) the Elias-Fano encoding.  This dictionary
appears as ``Exact Membership Tester 2'' in~\cite{carter1978exact}. We
refer to this implementation as the \emph{Pocket Dictionary}.

We view each element in the universe $[u']$ as a pair $(q,r)$, where
$q\in [B]$ and $r\in[u'/B]$ (we refer to $q$ as the \emph{quotient}
and to $r$ as the \emph{remainder}). Consider a multiset
$F\triangleq \set{(q_i,r_i)}_{i=0}^{n'-1}$.  The encoding of $F$ uses
two binary strings, denoted by $\hheader(F)$ and $\bbody(F)$, as
follows.  Let $n_q\triangleq |\set{i \in [f] \mid q_i=q}|$ denote the
number of elements that share the same quotient $q$.
The vector $(n_0,\ldots,n_{B-1})$ is stored in $\hheader(F)$ in unary
as the string $1^{n_0}\circ 0 \circ \cdots \circ 1^{n_{B-1}}\circ
0$. The length of the header is $B+n'$.
The concatenation of the remainders is stored in $\bbody(F)$ in
nondecreasing lexicographic order of $\set{(q_i,r_i)}_{i\in[n']}$. The
length of the body is $n'\cdot \log (u'/B)$. 
The space required is $B+n'(1+\log (u'/n))$ bits, which meets the
required space bound since $B=O(n')$.

Operations in a Pocket Dictionary can be executed in
constant time if the Pocket Dictionary fits in a single word (see~\cite{liu2020succinct}).
Rather than employing sophisticated manipulations in the (abstract)
classical RAM model, we propose to extend the classical RAM model as
follows. All computations over $O(1)$ words that can be excuted by
circuits of depth $O(\log w)$ and size $O(w^2)$ are considered to be
executable in constant time in the new RAM model.  This extension is
motivated by the fact that, one one hand, computations such as
comparison, addition, and multiplication take constant time in modern
CPUs.  On the other hand, circuits for these computations require
$\Omega(\log w)$ depth.  As for size, multipliers are implemented using
circuits of size $\Theta(w^2)$ (simply because all the partial
products are computed).

We note that operations over Pocket Dictionaries can be supported by
circuits of depth $\log w$ with $O(w\log w)$ gates.  Consider an
insertion operation of an element $(q,r)$.  Insertion is implemented
using the following steps (all implemented by circuits):
\begin{enumerate*}[label=(\roman*)]
\item Locate in the header the positions of $q$th zero and the zero
that proceeds it (this is a select operation). Let $j$ and $i$
denote these positions within the header.  This implies that
$\sum_{q'<q} n_{q'}=i-(q-1)$ and $n_q=j-i-1$.
\item Update the header by shifting the suffix starting in position
$j$ by one position and inserting a $1$ in position $j$.
\item Read $n_q$ remainders in the body, starting from position
$i-(q-1)$.  These remainders are compared in parallel with $r$, to
determine the position $p$ within the body in which $r$ should be
inserted (this is a rank operation over the outcomes of the
comparisons).
\item Shift the suffix of the body starting with position $p$ by $|r|$
bits, and copy $r$ into the body starting at position $r$. 
\end{enumerate*}

From a practical point of view, modern instruction sets support
instructions such as rank, select, and SIMD comparisons (shifts are
standard instructions)~\cite{Reinders2013,
  DBLP:conf/sigmod/PandeyBJP17, bender2017bloom}. Hence, one can
implement Pocket Dictionary operations in constant time using such
instruction sets.

\subsection{The Spare}
The spare is a dynamic dictionary that maintains (arbitrary)
multisets of cardinality at most $n_s$ from the universe $\UU$ with
the following guarantees: (1)~For every $\poly(n)$ sequence of
operations (insert, delete, or query), the spare does not overflow
whp. (2)~If the spare does not overflow, then every operation (query,
insert, delete) takes $O(1)$ time. (3)~The required space is
$O(n_s\log u)$ bits.\footnote{Since $n_s=o(n/\log n)$ and
  $\log u = O(\log n)$, the space consumed by the spare is $o(n)$.}

We propose to implement the spare by using a dynamic dictionary
on sets with constant time operations in which counters are appended
to elements. To avoid having to discuss the details of the interior
modifications of the dictionary, we propose a black-box approach that
employs a retrieval data structure in addition to the dictionary.

\begin{observation}\label{obs:spare} Any dynamic dictionary on sets of
cardinality at most $n_s$ from the universe $\UU$ can be used to
implement a dynamic dictionary on arbitrary multisets of
cardinality at most $n_s$ from the universe $\UU$.  This reduction
increases the space of the dictionary on sets by an additional
$\Theta(\log n_s + \log\log(u/n_s))$ bits per element and increases
the time per operation by a constant.
\end{observation}
  
\begin{proof}
The dictionary on multisets ($\MSDict$) can be obtained by employing
a dictionary on sets ($\Dict$) and 
the dynamic retrieval data structure ($\Ret$) of
Demaine~\etal~\cite{demaine2006dictionariis}.  The dictionary on
sets ($\Dict$) stores the support of the input multiset.  The
retrieval data structure ($\Ret$) stores as satellite data the
multiplicity of each element in the support. The space that $\Ret$
requires is $\Theta(n_s\log n_s + n_s\log\log(u/n_s))$, since the
satellite data occupies $\log n_s$ bits.

On membership queries, the dictionary accesses $\Dict$. 
When a new element $x$ is inserted, $\MSDict$ inserts $x$ in $\Dict$ and adds $x$ with 
satellite value $1$ to $\Ret$. In the case of insertions of duplicates or deletions,
$\Ret$ is updated to reflect the current multiplicity of the element. 
If upon deletion, the multiplicity of the element reaches $0$, the element is deleted from $\Dict$ and from $\Ret$. Since the $\Ret$ supports
operations in constant time, this reduction only adds $O(1)$ time to the operations on $\Dict$.

\end{proof}

To finish the description of the spare, we set
$n_s \triangleq n/(\log^3 n)$ and employ Obs.~\ref{obs:spare} with a
dynamic dictionary on sets that requires $O(n_s\log (u/n_s))$ bits and
performs operations in constant time whp~\cite{ dietzfelbinger1990new,
  dalal2005two,demaine2006dictionariis, arbitman2009amortized,
  arbitman2010backyard}. Under our definition of $n_s$ and since
$\log(u/n)=O(\log\log n)$, the spare then requires $o(n)$
bits. Moreover, the spare does not overflow whp (see Claim~\ref{claim:spare}).

\subsection{Hash Functions}
We consider three hash functions,
the bin index, quotient, and remainder, as follows: \footnote{One could define the
domain of the quotient function $q(x)$ and the remainder function
$r(x)$ to be $[u/m]$ instead of $\UU$.} (Recall that $n=m\cdot B$.)
\begin{align*}
h^{b}&: \UU \rightarrow \left[m\right] &\text{(bin index)}\\
q&: \UU \rightarrow \left[B\right] &\text{(quotient)}\\
r&: \UU\rightarrow \left[u/n\right]&\text{(remainder)}
\end{align*}

We consider three settings for the hash functions:
\begin{enumerate*}[label=(\roman*)]
\item Fully random hash functions. 
\item In the case that the dataset is a random multiset, the values of the
hash functions are taken simply from the bits of $x$.  Namely
$h^b(x)$ is the first $\log (n/B)$ bits, $q(x)$ is the next $\log B$
bits, and $r(x)$ is the last $\log(u/n)$ bits (to be more precise,
one needs to divide $x$ and take remainders). Since $x$ is chosen
independently and uniformly at random, the hash functions are
fully random.
\item Hash functions sampled from special distributions of hash
functions (with small representation and constant evaluation time).
\end{enumerate*}

\subsection{Functionality}
A $\query(x)$ is implemented by searching for $(q(x),r(x))$ in the
bin dictionary of index $\hb(x)$. If the pair is not found, the query
is forwarded to the spare. An $\ins(x)$ operation first attempts to insert $(q(x),r(x))$ in the
bin dictionary of index $\hb(x)$. If the bin dictionary is full, it forwards the insertion to the spare. 
A $\del(x)$ operation searches for the pair $(q(x),r(x))$ in the bin dictionary
of index $\hb(x)$ and deletes it (if found). Otherwise, it
forwards the deletion to the spare.

\subsection{Overflow Analysis\protect\footnote{ The proofs in this section assume that the
hash functions are fully random.  See
Section~\ref{sec:hash} for a discussion of special
families of hash functions.
}}\label{sec:ovf}

The proposed dictionary consists of two types of dictionaries: many
small bin dictionaries and one spare. The overflow of a bin dictionary is
handled by sending the element to the spare. Hence, for correctness to
hold, we need to prove that the spare does not overflow whp.

The first challenge that one needs to address is that the dictionary
maintains a dynamic multiset $D(t)$ (see~\cite{demaine2006dictionariis}).
Consider the insertion of an element $x$ at time $t$. If bin $h^b(x)$ is
full at time $t$, then $x$ is inserted in the spare. Now suppose that
an element $y$ with $h^b(y)=h^b(x)$ is deleted at time $t+1$. Then,
bin $h^b(x)$ is no longer full, and $x$ cannot ``justify'' the fact
that it is in the spare at time $t+1$ based on the present dynamic multiset
$D(t+1)$. Indeed, $x$ is in the spare due to ``historical reasons''.

The second challenge is that the events that elements are sent to the
spare are not independent. Indeed, if $x$ is sent to the spare at time
$t$, then we know that there exists a full bin. The existence of a
full bin is not obvious if the cardinality of $D(t)$ is small. Hence,
we cannot even argue that the indicator variables for elements being sent
to the spare are negatively associated. 

The following claim bounds the number of elements stored in the spare.
Using the same proof, one could show that the number of elements in
the spare is bounded by $n/(\log n)^c$ whp, for every constant $c$.
\begin{claim2}\label{claim:spare}
The number of elements stored in the spare is less than $n/\log^3 n$ with
high probability.
\end{claim2}
\begin{proof}
Consider a dynamic multiset $D(t)$ at time $t$.  To simplify notation,
let $\set{x_1,\ldots,x_{n'}}$ denote the multiset $D(t)$ (the
elements need not be distinct and $n'\leq n$).  Let $t_i$ denote
the time in which $x_i$ was inserted to $D(t)$.  Let $X_i$ denote
the random variable that indicates if $x_i$ is stored in the spare
(i.e., $X_i=1$ iff bin $h^b(x_i)$ is full at time $t_i$).  Our goal
is to prove that the spare at time $t$ does not overflow whp, namely:
\begin{align}\label{eq:spare t}
\Pr{\sum_{i=1}^{n'} X_i \geq \frac{n}{\log^3 n}} &\leq n^{-\omega(1)}\;.
\end{align}
The claim follows from Eq.~\ref{eq:spare t} by applying a
union bound over the whole sequence of operations.

To prove Eq.~\ref{eq:spare t}, we first bound the probability that a
bin is full. Let $\gamma \triangleq e^{-\delta^2\cdot B/3}$. Fix a
bin $b$, by a Chernoff bound, the probability that bin $b$ is full
is at most $\gamma$.  Indeed:
\begin{enumerate*}[label=(\roman*)]
\item Each element belongs to bin $b$ with probability 
$B/n$. Hence, the expected occupancy of a bin is $B$.
\item The variables $\set{h^b(x_j)}_{j=1}^{n'}$ are independent.
\item A bin is full if at least $(1+\delta)\cdot B$ elements belong
to it.
\end{enumerate*}

We overcome the problem that the random variables $\set{X_i}_i$ are not
independent as follows. Let $F_t$ denote the set of full bins at
time $t$.  If $|F_t|\leq 6\gamma m$, let $\hat{F}_t$
denote an arbitrary superset of $F_t$ that contains
$6\gamma m$ bins.  Note that it is unlikely that there exists
a $t$ such that $|F_t|>6\gamma m$.  Indeed, by linearity of
expectation, $\expectation{}{|F_t|}\leq\gamma m$.  By a
Chernoff bound, $\Pr{F_t> 6\gamma m} \leq 2^{-6\gamma m}$.

Define $\hat{X}_i$ to be the random variable that indicates if
$h^b(x_i)\in \hat{F}_{t_i}$. Namely, $\hat{X}_i=1$ if $x_i$ belongs
to a full bin or a bin that was added to $\hat{F}_{t_i}$. Thus,
$\hat{X}_i\geq X_i$.  The key observation is that the random
variables $\set{\hat{X}_i}_{i=1}^{n'}$ are independent and
identically distributed because the bin indexes $\set{h^b(x_i)}_i$
are independent and uniformly distributed.

The rest of the proof is standard.  Recall that $B=O(\log n)$ and
$\delta^2=\Theta\parentheses{\frac{(\log\log n)^2}{B}}$.  

Let $G_t$ denote the event that $F_t> 6\gamma m$.  Since
$\Pr{G_t} \leq 2^{-6\gamma m}\leq 2^{-\sqrt{n}}$, by a union
bound $\Pr{\bigcup_{i=1}^{n'} G_{t_i}}\leq n\cdot 2^{-\sqrt{n}}$.

The expectation of $\hat{X}_i$ is $6\gamma$ (conditioned on the event
$\bigcap_{i=1}^{n'} \overline{G}_{t_i}$). For a sufficiently large
$n$, we have $n/\log^3n \geq 6\cdot \gamma n$. By Chernoff bound
$\Pr{\sum_{i=1}^{n'} \hat{X}_i \geq \frac{n}{\log^3 n}\bigm|
  \bigcap_{i=1}^{n'} \overline{G}_{t_i}} < 2^{-n/\log^3 n}$.

We conclude that
\begin{align}
\Pr{\sum_{i=1}^{n'} \hat{X}_i \geq \frac{n}{\log^3 n}}  & \leq
\Pr{\bigcup_{i=1}^{n'} G_{t_i}} + \Pr{\sum_{i=1}^{n'} \hat{X}_i \geq \frac{n}{\log^3 n}\bigm| \bigcap_{i=1}^{n'} \overline{G}_{t_i}}\\
&\leq n\cdot 2^{-\sqrt{n}} + 2^{-n/\log^3 n} = n^{-\omega(1)}\;,
\end{align}
and Eq.~\ref{eq:spare t} follows.
\end{proof}

\section{Succinct Hash Functions}\label{sec:hash}
In this section we discuss how to replace the fully random hash
functions from Sec.~\ref{sec:rmsdic} with succinct hash functions (i.e., representation requires
$o(n)$ bits) that have constant evaluation time in the RAM model.  Specifically,
we describe how to select hash functions $h^b(x),r(x),q(x)$ for the
RMS-dictionary used for constructing the dynamic filter in the dense
case.

The construction proceeds in two stages and uses existing succint
constructions for highly independent hash
functions~\cite{siegel2004universal,dietzfelbinger2009applications}.
First, we employ a pseudo-random permutation
$\pi:\UUh\rightarrow \UUh$. The permutation $\pi$ can be either the
one-round Feistel permutation from~\cite{arbitman2010backyard} or the
quotient permutation from~\cite{demaine2006dictionariis}. We think of
$\pi$ as a pair of functions, i.e., $\pi(x)\triangleq (h_1(x),h_2(x))$
with $h_1(x) \in [M]$ and $h_2(x)\in [\frac{\uh}{M}]$. Hence $\pi$
induces a partition of the dynamic set into $M = n^{9/10}$
subsets, where the cardinality of each subset is at most
$n^{1/10}+n^{3/40}$ whp.  Each subset consists of the $h_2(x)$ values
of the elements $x\in\DD$ that share the same $h_1(x)$ value.

In the second step, we instantiate
the RMS-Dictionary separately for each subset. 
Each dictionary instance employs the same $k$-wise independent
hash function $f^b:[\uh/M]\rightarrow [m]$ with $k=n^{1/10}+n^{3/40}$.
We define $h^b(x) = f^b(h_2(x))$ to be the bin of $x$. 
From the perspective of each dictionary instantiation, $\hb(x)$ is sampled
independently and uniformly at random, so throughout the sequence of $\poly(n)$
operations, the spare does not overflow whp.

We now describe how the quotient $q(x)$ and the remainder $r(x)$ are
chosen.  We sample a $2$-independent hash function
$(f,g):[\uh/M]\rightarrow [B]\times [1/\eps]$.
Define $q(x)\triangleq f(h_2(x))$ and $r(x)\triangleq g(h_2(x))$.

\begin{claim2}
  Consider a filter based on an RMS-dictionary that employs the hash
  functions $h^b(x),q(x),r(x)$ described in this section.  Then the
  probability of a false positive event is bounded by $(1+o(1))\eps$.
\end{claim2}
\begin{proof}
  Consider a query $y$ that is not in the dataset $D(t)$ at time $t$.
  A false positive event can occur due to two subevents: (1)~$y$
  collides with an element $x$ mapped to bin $h^b(y)$,  bin
  has overflowed, and $x$ resides in the spare. The probability of
  this subevent is bounded by the probability that bin $h^b(y)$
  overflows. (2)~$y$ collides with an element $x$ stored in bin
  $h^b(y)$. 

  The probability that bin $h^b(y)$ overflows, namely, more than
  $(1+\delta)B$ elements are mapped to bin $h^b(y)$, is bounded by
  $2e^{-\delta^2 B/3}$ (see the proof of
  Claim~\ref{claim:spare}). Since $\delta^2 B=\Theta (\log\log n)^2$, and
  since $\log (1/\eps)=O(\log \log n)$, this probability is $o(\eps)$.

  Now assume that at most $(1+\delta)B$ elements in $\DD(t)$ were
  mapped to bin $h^b(y)$. What is the probability that $y$ collides
  with $x$ ($x$ must satisfy $h_1(x)=h_1(y)$ and $h^b(x)=h^b(y)$)?
  Since $(f,g)$ are selected from a family of $2$-independent hash
  functions, the probability of a collision with such an $x$ is
  bounded by $\frac{\eps}{B}$. A union bound over the bin, implies
  that a collision with an element in the bin occurs with probability
  at most  $(1+\delta)B\cdot \frac{\eps}{B}=(1+\delta)\cdot \eps$. The
  claim follows since $\delta=o(1)$.
  \end{proof}
\section{Dynamic Filter via Retrieval (Sparse Case)}\label{sec:sparse}
In this section, we present a space-efficient dynamic filter for the
case that $\log\parentheses{\frac{1}{\epsilon}}=\omega(\log \log n)$ (sparse case).\footnote{For 
an alternate construction which does not employ any reductions in the sparse case,
 we refer the reader to~\cite{bercea2019fullydynamic}.} We let $n$
denote an upper bound on the cardinality of a dynamic set over a
$\poly(n)$ sequence of insertions and deletions. We let $\UUh$ denote a
universe of cardinality $\uh$ that satisfies
$\uh=\poly(n)$. The
construction is based on a reduction to dynamic retrieval.
The construction relies on the fact that in dynamic retrieval structures
(e.g., \cite{demaine2006dictionariis}), the overhead per element is
$O(\log\log n)$. This overhead is $o(\log (1/\eps))$ in the sparse
case.

Dietzfelbinger and Pagh~\cite{dietzfelbinger2008succinct} formulate a
reduction that uses a static retrieval data structure storing $k$ bits
of satellite data per element to implement a static filter with false
positive probability $2^{-k}$.  The reduction is based on the
assumption that retrieval data structure is ``well behaved'' with
respect to negative queries. Namely, a query for
$x\in\UUh\setminus\DD$ returns either ``fail'' or the satellite data
of an (arbitrary) element $y\in \DD$. The reduction incurs no
additional cost in space and adds $O(1)$ extra time to the query
operations (to evaluate the fingerprint).  We note that the same
reduction can be employed in the dynamic case. Specifically, the
following holds:

\begin{observation}
  Assume access to a family of pairwise independent hash functions
  $h:\UUh\rightarrow [k]$.~\footnote{We note that Dietzfelbinger and
    Pagh~\cite{dietzfelbinger2008succinct} assume access to fully
    random hash functions. Pairwise independence suffices, however, as
    noted by~\cite{porat2009optimal}.} Then any dynamic retrieval data
  structure that stores $h(x)$ as satellite data for element $x$ can
  be used to implement a dynamic filter with false positive
  probability bounded by $2^{-k}$. The space of the resulting filter is the same
  as the space of the retrieval data structure and the time per
  operation increases by a constant (due to the computation of
  $h(x)$).
\end{observation}

The question of designing a space-efficient dynamic filter now boils down to:
\begin{enumerate*}[label=(\roman*)]
\item Choose a suitable dynamic retrieval data structure.
\item Determine the range of false positive probabilities $\eps$ for
which the reduction yields a space-efficient dynamic filter.
\end{enumerate*}
We resolve this question by employing the retrieval data structure of
Demaine~\etal~\cite{demaine2006dictionariis} in the sparse case.

\begin{claim2}\label{claim:sparse}
There exists a dynamic
filter in the sparse case that maintains a set of at most $n$ elements from the
universe $\UUh=[\uh]$, where $\uh=\poly(n)$ such that, for any $\eps$ such that
$\log(1/\eps)=\omega(\log\log n)$, the following hold: (1)~For every
sequence of $\poly(n)$ opertations (i.e, insert, delete, or query),
the filter does not overflow whp. (2)~If the filter does not
overflow, then every operation (query, insert, and delete) takes
$O(1)$ time. (3)~The required space is
$(1+o(1))\cdot n\log (1/\eps)$ bits. (4)~ For every query, the
probability of a false positive event is bounded by $\eps$.
\end{claim2}
\begin{proof}
  The dynamic retrieval data structure
  in~\cite{demaine2006dictionariis} uses a dynamic perfect hashing
  data structure for $n$ elements from the universe $\hUU$ of size $\hat{u}$
  that maps each element to a unique value in a given range $[n+t]$,
  for any $t>0$.~\footnote{ We note that one could also use the
    dynamic perfect hashing scheme proposed in
    Mortesen~\etal~\cite{mortensen2005dynamic} with similar space and
    performance guarantees.}  The space that the perfect hashing data
  structure occupies is
  $\Theta\parentheses{n\log\log\frac{\uh}{n} +
    n\log\frac{n}{t+1}}$. All operations are performed in $O(1)$ time
  and the perfect hashing data structure fails with $1/\poly(n)$
  probability over a sequence of $\poly(n)$ operations. A retrieval
  data structure can be obtained by allocating an array of $n+t$
  entries, each used to store satellite data of $k$ bits. The
  satellite data associated with an element $x$ is stored at the
  position in the array that corresponds to the hash code associated
  with $x$. The retrieval data structured obtained this way occupies
  $(n+t)\cdot k + \Theta\parentheses{n\log\log\frac{\uh}{n} +
    n\log\frac{n}{t+1}}$ bits.  It performs every operation in
  constant time and fails with probability $1/\poly(n)$ over a
  sequence of $\poly(n)$ operations.

For the filter construction, we set $k\triangleq \log(1/\eps)$ and $t\triangleq n/\log n$.
Since $\log\log(\uh/n) = O(\log\log n)$ and $\log(1/\eps) = \omega(\log\log n)$, the filter we obtain
occupies $(1+o(1))\cdot n\log(1/\eps)$ bits. It performs all operations in constant time and  does not
overflow whp over a sequence of $\poly(n)$ operations.
\end{proof}

\section{Proof of Theorem~\ref{thm:filter}}\label{sec:prooffilter}
The proof of Thm.~\ref{thm:filter} deals with the sparse case and
dense case separately. The theorem for the sparse case, in which
$\log(1/\eps)=\omega(\log\log n)$, is proven in Sec.~\ref{sec:sparse}.

The proof for the dense case employs the reduction in
Lemma~\ref{lemma:reduce dynamic} with the RMS-Dictionary construction
described in Sec.~\ref{sec:rmsdic}.  Let $\UU=[u]$ where $u=n/\eps$
denotes the universe of an RMS-Dictionary that can store a random
multiset of cardinality at most $n$. In this case, the assumption
that $\log(1/\eps)=O(\log\log n)$ translates into
$\log(u/n)=O(\log\log n)$.

The time per operation is constant because the RMS-dictionary supports
operations in constant time.  The space consumed by the RMS-Dictionary
equals the sum of spaces for the bin dictionaries and the spare. This
amounts to
$m\cdot n'\cdot \log (u'/n')+ m\cdot O(n') +o(n)\leq (1+\delta)\cdot
n\cdot \log (1/\eps) +O(n)$.  Since $\delta=o(1)$, the filter is
space-efficient, as required.  Finally, by Claim~\ref{claim:spare},
the spare does not overflow whp.

\section*{Acknowledgments}
We would like to thank Michael Bender, Martin Farach-Colton,
and Rob Johnson for introducing us to this topic and for interesting
conversations.
Many thanks to Tomer Even for helpful and thoughtful remarks.

\bibliography{main}
\end{document}